\documentclass[10pt]{amsart}
\usepackage{amsmath,amsfonts,amsthm,amssymb,eufrak,comment}
\usepackage{amscd}
\usepackage{times,fancyhdr}
\usepackage{array}
\usepackage{xy}
\usepackage{epsfig}
\usepackage{graphicx}
\usepackage{color}

\newcommand{\re}{{\text{\rm re}}}
\newcommand{\im}{{\text{\rm im}}}

\usepackage{euscript}

\newcounter{thecounter}
\numberwithin{thecounter}{section}
\newtheorem{lemma}[thecounter]{Lemma}
\newtheorem{proposition}[thecounter]{Proposition}
\newtheorem{theorem}[thecounter]{Theorem}
\newtheorem{thm}[thecounter]{Theorem}

\newcommand{\R}{{\mathbb{R}}}

\newcommand{\Z}{{\mathbb{Z}}}

\DeclareMathOperator{\rank}{rank}

\def\a{\alpha}
\def\b{\beta}

\newcommand{\realroots}{\Delta^\text{re}}
\begin{document}

\title{Root subsystems of rank 2 hyperbolic root systems}

\author{Lisa Carbone, Matt Kownacki, Scott H.\ Murray  and Sowmya Srinivasan}

\begin{abstract} Let $\Delta$  be a rank 2 hyperbolic root system. Then $\Delta$ has generalized Cartan matrix 
$H(a,b)= \left(\begin{smallmatrix} ~2 & -b\\
-a & ~2 \end{smallmatrix}\right)$
indexed by  $a,b\in\mathbb{Z}$ with $ab\geq 5$.  If $a\neq b$, then $\Delta$ is non-symmetric and is generated by one long simple root and one short simple root; whereas if  $a= b$,  $\Delta$ is symmetric and is generated by two long simple roots. We prove that  if $a\neq b$, then $\Delta$ contains an infinite family of symmetric rank 2 hyperbolic root subsystems $H(k,k)$ for certain $k\geq 3$, generated by either two short or two long simple roots. We also prove that $\Delta$ contains non-symmetric rank 2 hyperbolic root subsystems $H(a',b')$, for certain $a',b'\in\mathbb{Z}$ with $a'b'\geq 5$. One of our tools is a characterization of the types of root subsystems that are generated by a subset of roots. We classify these types of subsystems in rank 2 hyperbolic root systems.

\end{abstract}
\thanks{This research made extensive use of the Magma computer algebra system.}
\maketitle

\section{Introduction}

Let $\Delta$ be the root system of a rank 2 Kac--Moody algebra $\mathfrak{g}$. Then $\Delta$ has generalized Cartan matrix 
$H(a,b):= \left(\begin{smallmatrix} ~2 & -b\\
-a & ~2 \end{smallmatrix}\right)$
for some $a,b\in\mathbb{Z}$ with $ab\geq 5$. Let $S=\{\alpha_1,\alpha_2\}$ denote a basis of simple roots of $\Delta$.
If $a\neq b$, then $\Delta$ is non-symmetric and its base consists of a long simple root and a short simple root, whereas if  $a= b$,  $\Delta$ is symmetric and its base consists of two long simple roots. 

The root system $\Delta$ contains two types of roots; real and imaginary. The real roots of $\Delta$ are of the form $w\alpha_i$ for some $w\in W$, where $W$ is the Weyl group of the root system. In the rank 2 hyperbolic case, $W\cong D_{\infty}$, is the infinite dihedral group. The additional imaginary roots will not play a significant role in this work. The real roots are supported on the branches of a hyperbola in $\R^{(1,1)}$, with a pair of branches for each root length (Figure 1).

We are also motivated by the following question. If $\alpha$ and $\beta$ are real roots of a Kac-Moody group, then the commutator $[\chi_{\alpha},\chi_{\beta}]$  involves the root group corresponding to $\alpha+\beta$. The  commutator $[\chi_{\alpha},\chi_{\beta}]$ is then trivial if the sum $\alpha+\beta$ is not a root.  In order to determine the non--trivial commutators in  the  Kac--Moody group associated to $H(a,b)$, we may therefore reduce to the study of rank 2 root subsystems in a general Kac--Moody root system. This observation provides the setting for the current work.

As observed by Morita (\cite{Mor}, \cite{Mor2}), to determine the group commutators, it is also necessary to characterize the set  of positive $\Z$--linear combinations $(\mathbb{Z}_{>0}\alpha+\mathbb{Z}_{>0}\beta)\cap\Delta^{\re}(H(a,b))$  for  real roots $\alpha,\beta$ whose sum is a real root.   A similar question was answered in arbitrary Kac--Moody root systems by Billig and Pianzola ([BP]). We obtain an explicit description of this set for all rank 2 hyperbolic root systems.

In future work, we will use these results to determine the non--trivial group commutators and their structure constants ([CMW]).

We obtain proofs of the results stated by Morita (\cite{Mor})  that if $a$ and $b$ are both greater than one, then no sum of real roots can be a real root. It follows that the prounipotent subgroup corresponding to the positive real roots on a single branch of the hyperbola is commutative. When $a$ or $b=1$ we prove, as stated in \cite{Mor},  that the prounipotent subgroup generated by all the positive real short root groups is metabelian and the prounipotent subgroup generated by all the positive real long root groups is commutative. Our results in Sections~\ref{rootsums} and \ref{Rk2subsys} also cover the affine cases $H(2,2)$ and $H(4,1)$.

In order to make our results precise, we use two different concepts of a subsystem generated by a subset  $\Gamma$ of real roots: namely a subsystem $\Phi(\Gamma)$,  corresponding to a reflection subgroup of the Weyl group and consisting entirely of real roots; and $\Delta(\Gamma)$ consisting of all roots that can be written as an integral linear combination of elements of $\Gamma$. Such a $\Delta(\Gamma)$ subsystem contains both real and imaginary roots and corresponds to a certain subalgebra of the Kac--Moody algebra.

We have classified both kinds of subsystem inside a rank 2 infinite root system, and found that the two concepts of subsystem are equivalent in almost all cases:
\begin{thm} 
Let $\Delta$ be a  rank 2 infinite root system and let $\Gamma$ be a set of real roots which generate $\Delta$, that is, $\Delta(\Gamma)=\Delta$.
Then either $\Phi(\Gamma)$ is the set of all real roots in $\Delta$ or it is the set of all \emph{short} real roots in $\Delta$.
The second case occurs only if $a=1$ or $b=1$ and $\Phi(\Gamma)$ consists of short roots.
\end{thm}

Our classification also gives us the following result, which holds for either concept of subsystem:
\begin{thm}\label{infsubs} 
If $\Delta$ is a rank 2 hyperbolic root system, then $\Delta$ contains symmetric rank 2 hyperbolic root subsystems of type $H(k,k)$ for infinitely many distinct $k\geq 3$. 
If $\Delta$ is non-symmetric of type $H(a,b)$, then it
also contains non-symmetric rank 2 hyperbolic root subsystems of type $H(a\ell,b\ell)$ for infinitely many distinct $\ell\geq 2$.
\end{thm}

We also classify the rank 2 $\Phi$-subsystems as finite, affine or hyperbolic
systems:
\begin{thm} 
Let $\Delta$ be a rank 2 root system and let $\Gamma$ be a nonempty set of real roots in $\Delta$.
\begin{enumerate}
\item If $\Delta$ is finite, then $\Phi(\Gamma)$ is finite.
\item If $\Delta$ is affine of type $\widetilde{A}_1$, then $\Phi(\Gamma)$ has finite type $A_1$ or affine type $\widetilde{A}_1$.
\item If $\Delta$ is affine of type $\widetilde{A}_2^{(2)}$, then $\Phi(\Gamma)$ has finite type $A_1$, or affine type $\widetilde{A}_1$ or $\widetilde{A}_2^{(2)}$.
\item If $\Delta$ is hyperbolic, then $\Phi(\Gamma)$ has finite type $A_1$ or hyperbolic type.
\end{enumerate}
\end{thm}

We mention the following related works: This work was inspired by the papers \cite{Mor} and \cite{Mor2} where the results of interest were stated without proof. Feingold and Nicolai (\cite{FN}, Theorem 3.1)  gave a method for generating a subalgebra corresponding to a $\Delta(\Gamma)$--type root subsystem for a certain choice of real roots in any Kac--Moody algebra.

As in Section 4 of this paper, Casselman (\cite{C}) reduced the study of structure constants for Kac--Moody algebras to rank 2 subsystems. Some of our results in Section 4 overlap with Section 4 of \cite{C}.

Tumarkin (\cite{T}) gave a classification the sublattices of  hyperbolic root lattices of the same rank. However, he requires conditions on the possible angles between roots that exclude all but a finite number of rank 2 hyperbolic root systems. In contrast, for our intended application to Kac--Moody groups,  we require the explicit construction of the embedding of the simple roots of a subsystems into the ambient system, rather than just describing its root lattice.

The authors are very grateful to Chuck Weibel for his careful reading of the MSc thesis ([Sr]) of the fourth author. This research was greatly facilitated by experiments carried out in the computational algebra systems Magma (\cite{BCFS}) and Maple (\cite{M}).

\section{Real Roots}
Let $A=H(a,b)$ be the $2\times 2$ generalized Cartan matrix\footnote{This is the transpose of the generalized Cartan matrix $A$ in [ACP].}
$$
A=H(a,b) = (a_{ij})_{i,j=1,2} =
\begin{pmatrix}
2 & -b \\
-a & 2
\end{pmatrix}
$$
for positive integers $a,b$, with Kac--Moody
algebra $\mathfrak{g}=\mathfrak{g}(A)$, root system
$\Delta=\Delta(A)$, and Weyl group $W=W(A)$. When $ab<4$, $A$ is positive definite and so $\Delta$ is finite.
When $ab=4$, $A$ is
positive semi-definite but not positive definite and so $\Delta$ is affine. When $ab>4$, $A$ is indefinite but every proper generalized Cartan submatrix is
positive definite, and so $A$ is  hyperbolic. 
Without loss of generality, we assume that $a\ge b$.

Let $S=\{\alpha_1,\alpha_2\}$ denote a basis of simple roots of $\Delta$.
We have the simple root reflections
$$w_j(\alpha_i)=\alpha_i-a_{ij}\alpha_j$$
for $i=1,2$ 
with matrices with respect to $S$
$$[w_1]_S = \left(\begin{array}{rr} -1&b\\0&1\end{array}\right),\qquad
[w_2]_S = \left(\begin{array}{rr}  1&0\\a&-1\end{array}\right).$$
The Weyl group $W=W(A)$ is the group generated by the simple root reflections $w_1$ and $w_2$.

Let
$$B = B(a,b)=\left(\begin{array}{rr}  2a/b&-a\\-a&2\end{array}\right)$$
be a symmetrization of $A$.
This defines the symmetric bilinear form 
$(u,v) = [u]_S^TB[v]_S$ and quadratic form $||u||^2=(u,u)$, which are  preserved under the action of $W$.
So the set of real roots is $$\Delta^\re=W\a_1\cup W\a_2.$$
Since $\Delta$ is finite, affine, or hyperbolic, the set of imaginary roots is 
$$\Delta^\text{im}=\{\alpha \in \mathbb{Z}\a_1+\mathbb{Z}\a_2 \mid \a\ne0\text{ and }||\a||^2 \le 0\}.$$

A diagram of the hyperbolic root system $H(5,1)$ is  given in Figure 1.

\begin{figure}[h]
\begin{center}
{\includegraphics[width=2.6in]{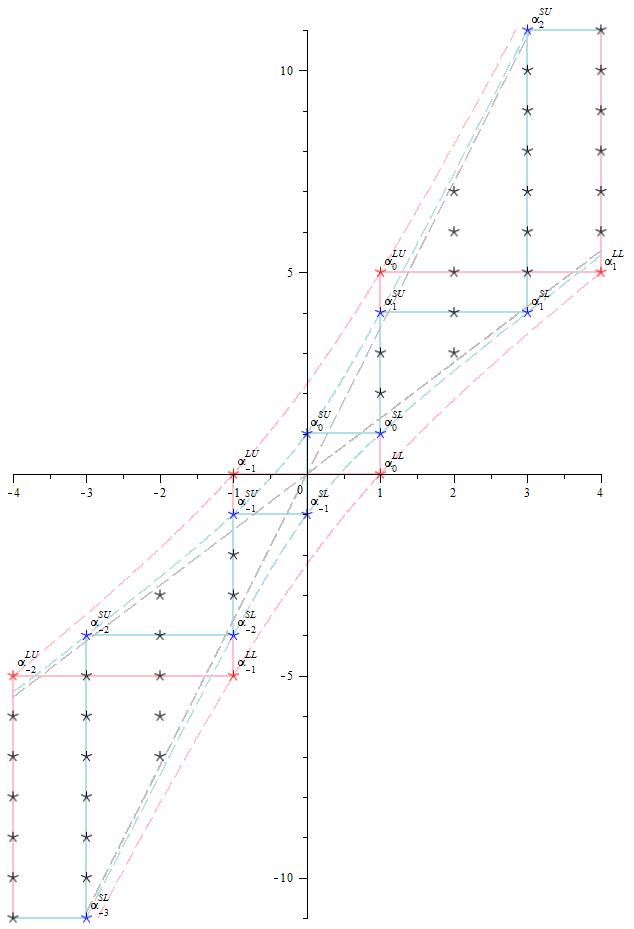}}
\caption{Root system of type $H(5,1)$}
\label{H51}
\end{center}
\end{figure}

Every root $\alpha\in\Delta$ has an expression of the form $\alpha= k_1\alpha_1+k_2\alpha_2$ where the $k_i$ are either all $\geq 0$, in which case $\alpha$ is called {\it positive}, or all $\leq 0$, in which case $\alpha$ is called {\it negative}. The positive roots are denoted $\Delta_+$, the negative roots $\Delta_-$.

Now $||\a_1||^2=2a/b$ and $||\a_2||^2=2$. 
So all real roots $x\alpha_1+y\alpha_2$ in the orbit $W\a_1$ satisfy
$$ {a}x^2 -abxy+by^2 = b,$$
and all real roots $x\alpha_1+y\alpha_2$ in the orbit $W\a_2$ satisfy
$$ {a}x^2 -abxy+by^2 = a.$$
These curves are displayed in Figures~1--3 as blue (resp. red) dotted lines. These curves are elliptical for finite systems, straight lines for affine systems, and hyperbolas for hyperbolic systems.
If $\Delta$ is nonsymmetric ($a>b$) the roots in $W\a_1$ are called \emph{long} and the roots in $W\a_2$ are called \emph{short}.
If $\Delta$ is symmetric ($a=b$) then all roots are considered to be \emph{long}.
Note that (with the exception of $A_2$), the real roots fall into two distinct orbits under the action of $W$. The diagrams use red for the orbit of $\alpha_1$, blue for the orbit of $\alpha_2$, and black for the imaginary roots. The horizontal lines indicate the action of $w_1$ while the vertical lines indicate the action of $w_2$.





For $j\in\mathbb{Z}$, we define
\begin{align*}
\alpha^{LL}_j &:= (w_1w_2)^j\alpha_1,&\alpha^{LU}_j &:= (w_2w_1)^jw_2\alpha_1\\
\alpha^{SU}_j &:= (w_2w_1)^j\alpha_2, &\alpha^{SL}_j &:= (w_1w_2)^jw_1\alpha_2.
\end{align*}
All real roots are given by these four sequences. If $ab\ge4$, then these are all distinct, and  a root is positive  if and only if $j \geq 0$. 
The roots for $j=-1,0,1$ are given in Table~\ref{T-rts}.

\begin{table}[h]
\begin{center}
\begin{tabular}{ c ||l|l|l  }
  $j$&  $-1$                              & $0$ & $1$ \\\hline
 $\alpha^{LL}_j$              & $-\alpha_1-a\alpha_2$ & $\alpha_1$ & $(ab-1)\alpha_1+a\alpha_2$                    \\
 $\alpha^{LU}_j $              & $-\alpha_1$ & $\alpha_1+a\a_2$ & $(ab-1)\alpha_1+a(ab-2)\alpha_2$                    \\
$\alpha^{SU}_j $               &$-b\alpha_1-\alpha_2$ & $\alpha_2$ &  $b\alpha_1+(ab-1)\alpha_2$ \\
$\alpha^{SL}_j$       &$-\alpha_2$ & $b\alpha_1+\alpha_2$ &  $b(ab-2)\alpha_1+(ab-1)\alpha_2$ 
\end{tabular}
\end{center}
\caption{Examples of real roots}\label{T-rts}
\end{table}


The following lemma characterizes the real roots in terms of recursive sequences $\eta_j$ and $\gamma_j$. Values of these sequences for small $j$ are given in Table~\ref{smallm}.
\begin{table}[h]
{\tiny
$$\begin{array}{r|rr}
j & \gamma_j & \eta_j \\\hline
0 & 0 & 1 \\
1 & 1 & ab - 1 \\
2 & ab - 2 & a^2b^2 - 3ab + 1 \\
3 & a^2b^2 - 4ab + 3 & a^3b^3 - 5a^2b^2 + 6ab - 1 \\
4 & a^3b^3 - 6a^2b^2 + 10ab - 4 & a^4b^4 - 7a^3b^3 + 15a^2b^2 - 10ab + 1 \\
5 & a^4b^4 - 8a^3b^3 + 21a^2b^2 - 20ab + 5 & a^5b^5 - 9a^4b^4 +
28a^3b^3 - 35a^2b^2 + 15ab - 1 \\
6 & a^5b^5 - 10a^4b^4 + 36a^3b^3 - 56a^2b^2 + 35ab - 6 & a^6b^6 -
11a^5b^5 + 45a^4b^4 - 84a^3b^3 + 70a^2b^2 - 21ab + 1 \\
7 & a^6b^6 - 12a^5b^5 + 55a^4b^4 - 120a^3b^3 + 126a^2b^2 - 56ab + 7 &
a^7b^7 - 13a^6b^6 + 66a^5b^5 - 165a^4b^4 + 210a^3b^3 - 126a^2b^2 +
28ab - 1 \\
8 & a^7b^7 - 14a^6b^6 + 78a^5b^5 - 220a^4b^4 + 330a^3b^3 - 252a^2b^2
+ 84ab - 8 & a^8b^8 - 15a^7b^7 + 91a^6b^6 - 286a^5b^5 + 495a^4b^4 -
462a^3b^3 + 210a^2b^2 - 36ab + 1 
\end{array}$$}
\caption{Values of $\eta_j$ and $\gamma_j$ for small $j$}
\label{smallm}
\end{table}
\begin{lemma}\label{acp} ([ACP], Lemmas 3.2 and 3.3)$\;$ For all integers $j$,
\begin{align*}
\alpha_j^{LL}&= \eta_{j}\alpha_1+a\gamma_{j}\alpha_2,& 
\alpha_j^{LU}&= \eta_{j}\alpha_1+a\gamma_{j+1}\alpha_2,\\
\alpha_j^{SU}&= b\gamma_j\alpha_1+\eta_{j}\alpha_2,& 
\alpha_j^{SL}&= b\gamma_{j+1}\alpha_1+\eta_{j}\alpha_2,
\end{align*}
where
\begin{enumerate}
\item $\gamma_0=0$, $\gamma_1=1$, $\eta_0=1$, $\eta_1=ab-1$;
\item $\eta_{j}=ab\gamma_{j}-\eta_{j-1}$;
\item $\gamma_{j}=\eta_{j-1}-\gamma_{j-1}$;
\item both sequences $X_j=\eta_j$ and $\gamma_j$ satisfy the recurrence relation
$$X_j=(ab-2)X_{j-1}-X_{j-2}.$$
\end{enumerate}
\end{lemma}
Note that these are both generalized Fibonacci sequences provided that $ab>4$. In particular, $\gamma_j$ is the Lucas sequence with parameters $P=ab-2,Q=1$.

The following is a useful  lemma giving negatives of roots:
\begin{lemma}\label{L-neg} For all $j\in\mathbb{Z}$, $\gamma_{-j}=-\gamma_{j}$ and $\eta_{-j}=-\eta_{j-1}$. Also
$$
-\alpha^{LL}_j = \alpha^{LU}_{-j-1}, \quad -\alpha^{LU}_j = \alpha^{LL}_{-j-1},\quad 
-\alpha^{SU}_j = \alpha^{SL}_{-j-1}, \quad -\alpha^{SL}_j = \alpha^{SU}_{-j-1}.
$$
\end{lemma} 

\section{Sums of real roots}\label{rootsums}
Let $\Delta$ be an infinite rank 2 root system of type $H(a,b)$ with $a\ge b$ and $ab\ge4$.
In this section we determine all real roots $\alpha,\beta\in\Delta$ for which  $\alpha+\beta$ is also a real root.

We will split our analysis into two cases: that in which $a\geq b >1$, and that in which $a>b=1$. We find in the first case that the sum of two real roots is never a real root, and in the second case that there are certain $\beta\in\realroots$ so that $\beta\pm \alpha_i \in \realroots$.

\subsection{The case $a\ge b>1$}

\begin{lemma}\label{staircase} If $a\ge b>1$, then
\begin{eqnarray*}
0=b\gamma_0<\eta_0<b\gamma_1<\eta_1<b\gamma_2<\cdots,\\
0=a\gamma_0<\eta_0<a\gamma_1<\eta_1<a\gamma_2<\cdots.
\end{eqnarray*}
In fact the gaps between sequence elements are nondecreasing, that is, for $j\ge0$,
\begin{align*}
 \eta_{j+1}-b\gamma_{j+1}&\ge b\gamma_{j+1}-\eta_j\ge\eta_j-b\gamma_j,\\
 \eta_{j+1}-a\gamma_{j+1}&\ge a\gamma_{j+1}-\eta_j\ge\eta_j-a\gamma_{j}.
\end{align*}
\end{lemma}
\begin{proof}
 To see that the gaps in the sequences are nondecreasing, we apply Lemma~\ref{acp} as follows: 
\begin{align*}
\eta_{j+1}-b\gamma_{j+1}&= (a-1)b\gamma_{j+1}-\eta_j
\ge b\gamma_{j+1}-\eta_m= (b-1)\eta_j-b\gamma_{j}\ge \eta_j-b\gamma_{j}. 
\end{align*} 
The other result is similar.
\end{proof}

The inequalities in Lemma~\ref{staircase} show that the real roots have the "staircase pattern'' shown in Figure~\ref{bgt1}.
\begin{figure}[h]
\begin{center}
{\includegraphics[width=3.4in]{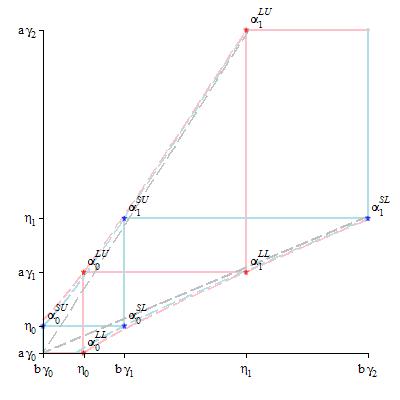}}
\caption{The positive real roots for $H(a,b)$ with $a\ge b>1$}
\label{bgt1}
\end{center}\end{figure}

\begin{proposition}\label{sumbne1}
If $a\ge b>1$ and $\a,\beta\in\Delta^\re$, then $\a+\b\notin\Delta^\re$.
\end{proposition}
\begin{proof}
We can write $\a=w\a_i$ for $i=1$ or 2, and some $w\in W$. 
We may also replace $\beta$ by $w^{-1}\beta$. Thus we wish to  determine the $\beta\in\Delta^\text{re}$ for which $\alpha_i+\beta\in\Delta^\re$.

Replacing $\beta$ by $-\beta$ if $\beta\in\Delta^\re_-$, we wish to determine the $\beta\in\Delta_+^\text{re}$ for which $\alpha_i\pm\beta\in\Delta^\re$.
Thus we may take $\a=\pm\a_i$ and $\b\in\Delta^\re_+$.
 
From Figure~\ref{bgt1}, it is clear that $\b\pm\a_i\in\Delta^\re$ only when  one of the differences $\eta_{j+1}-\eta_j$ or
$\gamma_{j+1}-\gamma_j$ equals 1. By Lemma~\ref{staircase}, this cannot occur.
\end{proof}

\subsection{The case $a>b=1$}
This case  is considerably more intricate. 

We define real functions
$\Psi_\pm(x) := \frac12\left((x-2)\pm\sqrt{x(x-4)}\right)$,
for which $\psi_\pm=\Psi_\pm(ab)$ are the characteristic roots of the recurrence equation in
Lemma~\ref{acp}(iv).

The following lemma gives a bound on these parameters. 
\begin{lemma}\label{bound} If $ab>4$, then $\psi_+>2.61$ and $0<\psi_-< 0.39$. 
\end{lemma}
\begin{proof}

We can  find $$\Psi'_\pm(x)=\frac12\left(1\mp\frac{x-2}{\sqrt{x(x-4)}}\right)$$
and so $\Psi_+$ is increasing for $x\ge5$ and $\Psi_-$ is positive and decreasing for $x\ge5$.
Hence $\psi_+\ge\Psi_+(5)>2.61$ and $0<\psi_-\le\Psi_-(5)<0.45.$
\end{proof}
Define
$$ \lambda := \frac{\psi_+}{\psi_+-1}, \qquad \mu:=\frac{1}{\sqrt{ab(ab-4)}}.$$
The following lemma is an easy consequence of Lemma~\ref{bound} and shows that the sequences $\eta_j$ and $\gamma_j$ are each within a small constant of being exponential with base $\psi_+$.
\begin{lemma} If $ab>4$, then, for $j\ge0$,
\begin{align*}
\lambda\psi_+^j-1.62<\eta_j&<\lambda\psi_+^j,\\
\mu\psi_+^j-0.45<\gamma_j&<\mu\psi_+^j,
\end{align*}
where $\psi_+>2.61$, $1 < \lambda < 1.62$, and $0 < \mu < 0.45$.
\end{lemma}

The following lemma shows that the roots have the "staircase pattern'' shown in Figure~\ref{beq1}.

\begin{lemma} \label{staircase2} If $a>4$ and $b=1$, then
\begin{align*}
0&=\gamma_0<\eta_0=\gamma_1<\gamma_2<\eta_1<\gamma_3<\eta_2<\cdots,\\
0&=a\gamma_0<\eta_0<\eta_1<a\gamma_1<\eta_2<a\gamma_2<\eta_3<a\gamma_3<\cdots.
\end{align*}
\end{lemma}
\begin{proof} For the first two inequalities, the previous lemmas show  that $0<\mu<0.45$ and $1 <\lambda<1.62$. Also, we have: 
$$
\lambda\psi^+-a\mu = \frac{\psi_+^2}{\psi_+-1} - \frac{a}{\sqrt{a(a-4)}} \geq 2.62-0.45 >2 
$$ and 
\begin{align*}
\mu\psi_+^2-\lambda  = \frac{\psi_+^2}{\sqrt{a(a-4)}} - \frac{\psi_+}{\psi_+-1} 
> 0.45\psi_+^2 - 0 > 3. 
\end{align*}

For $j\geq 1$ we have $\gamma_{j+1} = \eta_j - \gamma_j < \eta_j $. Similarly for $j\ge0$ we have $\eta_{j+1} = a\gamma_{j+1}-\eta_{j} < a\gamma_{j}$. Now
\begin{align*}
\eta_j-a\gamma_{j-1} &> \lambda\psi_+^j-1.62 -a\mu\psi_+^{j-1}\\
&= (\lambda\psi^+-a\mu)\psi^{j-1}-1.62\\
&\ge 2\psi^{j-1}-1.62>0 
\end{align*}
and so  $a \gamma_j < \eta_{j+1} $. And
\begin{align*}
\gamma_{j+1}-\eta_{j-1}&> \mu\psi_+^{j+1}-0.45-\lambda\psi_+^{j-1} \\
&= \psi_+^{j-1}(\mu\psi_+^2-\lambda)-0.45 \\
&\geq 3\psi_+^{j-1} - 0.45 > 0 
\end{align*}
and so $\gamma_{j+1}>\eta_{j-1}$. \end{proof}
\begin{figure}
\begin{center}
{\includegraphics[width=3.4in]{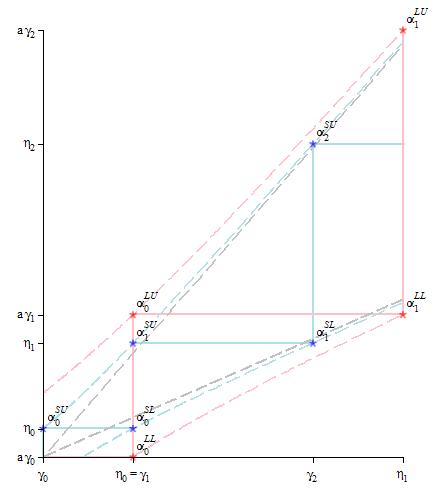}}
\caption{The positive real roots for $H(a,1)$ with $a>4$}
\label{beq1}
\end{center}\end{figure}
%

 We now use the above results to determine the $\beta\in\realroots$ for which $\beta\pm \alpha_i\in\realroots$, for $i=1,2$. 

\newpage
\begin{theorem}
\label{beq1}
If $a\ge4$, $b=1$ and $\beta\in\Delta_+^\text{re}$ then
\begin{enumerate}
\item $\beta+\alpha_1\in\Delta^\text{re}$ if and only if  $\beta=\alpha_2$;
\item $\beta-\alpha_1\in\Delta^\text{re}$ if and only if  $\beta=\alpha_1+\alpha_2$;
\item $\beta+\alpha_2\in\Delta^\text{re}$ if and only if  $\beta=\alpha_1$ or $\alpha_1+(a-1)\alpha_2$;
\item $\beta-\alpha_2\in\Delta^\text{re}$ if and only if  $\beta=\alpha_1+\alpha_2$ or $\alpha_1+a\alpha_2$.
\end{enumerate}
\end{theorem}

\begin{proof} For (i),  note that if $\beta = \alpha_2$, then $\beta+\alpha_1 = \alpha_1+\alpha_2 = \alpha_0^{SL}$, by Table 1.  Thus $\beta+\alpha_1$ is  a real root. 

Conversely, let $\beta\in\realroots$ such that $\beta+\alpha_1 \in\realroots$. 
First suppose that $\beta=\alpha_j^{SU}$ for some $j$. By Lemma~\ref{acp}, we have 
$\beta_j = \gamma_j\alpha_1 + \eta_j\alpha_2.$ Then $$\beta+\alpha_1 = (\gamma_j+1)\alpha_1+\eta_j\alpha_2.$$
If $\beta+\alpha_1$ is a long root on a lower branch, then again by Lemma~\ref{acp} we have 
$\beta+\alpha_1 = \eta_k \alpha_1 + a\gamma_k\alpha_2$ for some $k$. 
Then we must have $\gamma_{j}+1 = \eta_k$ and $\eta_j = a\gamma_k$, but by Lemma~\ref{staircase}  there are no $j,k$ such that $\eta_j = a\gamma_k$. Similarly, there are no $j,k$ such that $\gamma_j+1 = \eta_k$ and 
$\eta_j = a\gamma_{k+1}$, so $\beta+\alpha_1$ cannot be a long root on an upper branch. 
If $\beta+\alpha_1$ is a short root on an upper branch, then Lemma~\ref{acp} implies that $\gamma_j+1 = \gamma_k $ and $\eta_j = \eta_k$, but again, 
no such $j,k$ can exist: if $\beta+\alpha_1$ 
is a short root on an upper branch, then we must have $j,k$ so that
$$ \gamma_j+1 = \gamma_{k+1} \,\,\, \text{and} \,\,\, \eta_j = \eta_k.$$
The second of these conditions implies that $j=k$. Then by the first of these conditions we have $\gamma_{j+1}-\gamma_j = 1$. Then $j$ must be $0$. So 
$$\beta+\alpha_1 = \alpha_0^{SL} = \alpha_1 + \alpha_2,$$
so $\beta$ is $\alpha_2$. Similarly we may check that if $\beta = \alpha_{j}^{LL}$, $\beta=\alpha_j^{LU}$, or $\beta=\alpha_j^{SL}$ for some $j>0$, then $\beta+\alpha_1 \not\in\realroots$. 

Following similar reasoning, we can check that the second, third and fourth claims hold. \end{proof}

By a straightforward case-by-case argument, we can now prove the following result about lengths of sums of roots.

\begin{thm} \label{sums}
Let $\Delta$ be an infinite rank 2 root system.
\begin{enumerate}
\item If $\alpha,\beta,\alpha+\beta\in\Delta^{\re}$ with $\alpha$ and $\beta$ short, then $\alpha+\beta$ is long.
\item If $\alpha,\beta,\alpha+\beta\in\Delta^{\re}$ with $\alpha$ short and $\beta$ long, then $\alpha+\beta$ is short.
\item If $\alpha,\beta\in\Delta^{\re}$ with $\alpha$ and $\beta$ long, then $\alpha+\beta\notin\Delta^{\re}$.
\end{enumerate}
\end{thm}
We note that (i) and (iii) are not true in finite root systems of type $A_2$ or $G_2$. However there is a slightly weaker result that holds in any symmetrizable system:
\begin{thm}
Let $\Delta$ be a symmetrizable root system and suppose $\alpha,\beta,\alpha+\beta\in\Delta^{\re}$.
\begin{enumerate}
\item If $||\alpha||^2=||\beta||^2$, then $||\alpha+\beta||^2=a||\alpha||^2$ for some positive integer $a$.
\item If $||\alpha||^2\ne||\beta||^2$, then $||\alpha+\beta||^2=\min(||\alpha||^2,||\beta||^2)$.
\end{enumerate}
\end{thm}
\begin{proof}
We only need to consider the rank 2 subsystem $\mathbb{Z}\{\a,\b\}\cap\Delta$. These results are easily shown to be true if the subsystem has finite type $A_2$, $B_2$, or $G_2$, and they follow from the previous proposition if the subsystem is infinite.
\end{proof}


\section{Subsystems}\label{Rk2subsys}
Root systems can be used to describe three different structures: Coxeter groups, Kac--Moody algebras and Kac--Moody groups. 
These three structures lead to two different concepts of 
subsystem, since the Lie correspondence ensures that Kac--Moody algebras and  groups give the same subsystems.

In this section, we describe two distinct types of root subsystem and show  that the two concepts usually coincide, but not always. We also classify all subsystems of infinite rank 2 root systems.

 Suppose $\Delta$ is a symmetrizable root system with simple roots $\Pi=\{\alpha_1,\dots ,\alpha_{\ell}\}$.
Let $W=W(\Delta)$ be the Weyl group and  let $\Delta^\re=W\Pi$ denote the real roots.

For  $\Gamma\subseteq\Delta^\re$,  the  {\it reflection subgroup} generated by $\Gamma$ is defined as
  $$W_\Gamma = \langle w_\alpha : \alpha \in \Gamma \rangle.$$
Then $W_\Gamma$ is also a Coxeter group. 
We define
 $$\Phi(\Gamma)= W_\Gamma\Gamma,$$ 
that is, the closure of $\Gamma$ under the action of $W_\Gamma$.
We call this a \emph{$\Phi$-subsystem} (also noted in [C], Proposition 7). Note that a $\Phi$-subsystem consists entirely of real roots.

Let $\alpha$ be any real root. Then there is a corresponding pair of root vectors $x_\alpha$ and $x_{-\alpha}$ in $\mathfrak{g}=\mathfrak{g}(\Delta)$ which generate a subalgebra
isomorphic to $\mathfrak{sl}_2$.
We denote this subalgebra by $\mathfrak{sl}_2(\alpha)$.
Now let $\Gamma\subseteq \Delta^{\re}$. We may define the fundamental Kac--Moody subalgebra corresponding to $\Gamma$ to be 
  $$\mathfrak{g}_\Gamma = \langle \mathfrak{h}, \mathfrak{sl}_2(\alpha) : \alpha \in \Gamma \rangle.$$

 Then $\mathfrak{g}_\Gamma$ is a Kac--Moody algebra and its root system is
$$\Delta(\Gamma)= \mathbb{Z}\Gamma \cap \Delta,$$
that is the set of all roots in $\Delta$ that can be written as an integer linear combination of elements of $\Gamma$.
We call this a \emph{$\Delta$-subsystem}. The Kac--Moody subalgebra of \cite{FN}, Theorem 3.1 is of this type. We also define $\Delta^\re(\Gamma)=\mathbb{Z}\Gamma\cap\Delta^\re$.

\subsection{Subsystems corresponding to submatrices}
In this section we discuss one of the easiest ways to construct subsystems, and show that $\Delta$- and $\Phi$-subsystems coincide in this case.
Let $\mathfrak{g}$ be a Kac--Moody algebra with generalized Cartan matrix $A=(a_{ij})_{i,j\in I}$, $I=\{1,2,\dots ,\ell\}$, Cartan subalgebra $\mathfrak{h}$ of dimension $2\ell-\rank(A)$, simple roots $\Pi=\{\alpha_1,\dots,\alpha_{\ell}\}\subseteq \mathfrak{h}^{\ast}$ and simple 
coroots $\Pi^{\vee}=\{\alpha_1^{\vee},\dots,\alpha_{\ell}^{\vee}\} \subseteq\mathfrak{h}$. Let $Q$ denote the root lattice of $\mathfrak{g}$ and let $\mathfrak{g}=\mathfrak{h}\oplus\left(\bigoplus_{\alpha\in Q\backslash\{0\}}\mathfrak{g}^{\alpha}\right)$ denote the root space decomposition. Let $\Delta$ denote the set of all roots.

\medskip 
\noindent  Let $B=(a_{ij})_{i,j\in K}$ be a submatrix of $A$ for some $K\subseteq I$ with $|K|=\ell_0$. Let $\mathfrak{h}(B)$ be a subspace of $\mathfrak{h}$ of dimension $2\ell_0-\rank(A_0)$ containing $\Pi(B)^{\vee}=\{\alpha_i^{\vee} \mid i\in K\},$ and such that $\Pi(B)=\{\alpha_i|_{\mathfrak{h}(B)^{\ast}} \mid i \in K\}$ is linearly independent. 
Set $Q_0=\bigoplus_{i\in K} \mathbb{Z}\alpha_i$. Then 
$$\mathfrak{g}_0\cong \mathfrak{h}(B)\oplus\left(\bigoplus_{\alpha\in Q_0\backslash\{0\}} \mathfrak{g}^{\alpha}\right)$$
is the Kac--Moody algebra of $B$ with Cartan subalgebra $\mathfrak{h}(B)$, simple roots $\Pi(B)$ and simple coroots $\Pi(B)^{\vee}$ ([K], Exercise 1.2).
We identify $Q_0\subseteq Q\subset \mathfrak{h}^\ast$ with $\Z\Pi(B)\subset\mathfrak{h}(B)^\ast$ in the obvious way.







\begin{proposition} (Proposition 6, [Mo]) $\;$Let $A=(A_{ij})_{i,j\in I}$ be a generalized Cartan matrix with $I=\{1,2,\dots ,\ell\}$. Let 
$K\subset I$, $K\neq \varnothing$. Let $\Pi(A)=\{\alpha_1,\dots ,\alpha_{\ell}\}$ be the simple roots of the root system $\Delta(A)$. Let $B=(A_{ij})_{i,j\in K}$. Let $\Delta(B)$ denote the root system corresponding to $B$. Let 
$\Pi(B)=\Pi(A)\cap\Delta(B)$. Then $\Pi(B)\neq\varnothing$ and $\Delta(B)$  has the properties:
$$\Delta(B)=\Delta\cap \Z\Pi(B)$$
where $\Z\Pi(B)$ denotes all integral linear combinations of $\Pi(B)$
and
\begin{align*}
\Delta(B)^{\re}&=\Delta^{\re}\cap \Z\Pi(B),\\
\Delta(B)^{\im}&=\Delta^{\im}\cap \Z\Pi(B).
\end{align*}
\end{proposition}

\begin{proposition}  Using the notation above, define $\Phi(B)=W_{\Pi(B)}(\Pi(B))$. Then 
$$\Delta(B)^{\re}=\Phi(B).$$
That is, $\Phi(B)$ and $\Delta(B)^{\re}$ subsets coincide. 
\end{proposition}

\begin{proof}  Our claim is that
$$\Delta^{\re}\cap \Z\Pi(B) = W_{\Pi(B)}(\Pi(B)).$$
The inclusion $\Delta(B)^{\re}\subseteq \Phi(B)$ is clear. To prove the reverse inclusion, let $\alpha\in \Phi(B)$. Then $\alpha\in \Delta^{\re}$ and $\alpha\in \Z\Pi(B)$ by definition. Hence $\Phi(B)\subseteq \Delta(B)^{\re}$.
\end{proof}

The following lemma establishes a useful property of $\Delta(B)$ subsystems.

\begin{lemma} $\Delta(B)^{\re}$ subsets are closed  under taking integral sums of the simple roots corresponding to $B$. 
\end{lemma}
\begin{proof} We have $\Delta(B)^{\re}=\Delta^{\re}\cap \Z\Pi(B)$, but since $B$ is a subsystem arising from a submatrix of the generalized Cartan matrix, $B$ has an associated root lattice $Q_0=\bigoplus_{i\in K} \mathbb{Z}\alpha_i$ which is closed under taking integral sums. (See also Section 4, in particular Lemma 6, of [C]).
\end{proof}

Since $\Phi(B)$ and $\Delta(B)^{\re}$ subsets coincide, we claim  that $\Phi(B)$ subsets have this property as well. 

\begin{proposition} Using the notation above,  $\Phi(B)$ subsets  are closed  under taking integral sums of the simple roots corresponding to $B$.

\end{proposition}

\begin{proof} Let $\alpha,\beta\in \Phi(B)$.
We recall that $$-p\a+\beta,\ \dots,\ \b-\a,\ \beta, \ \a+\b, \ \dots ,\ q\a+\beta$$ is the $\alpha$--string through $\beta$. We claim that for $s,t\in\Z$, $s\alpha+t\beta\in \Z\Pi(B)$. Let
$$\alpha=\sum_{i\in K} a_i\alpha_i \quad\text{and  \ }
\beta=\sum_{i\in K} b_i\alpha_i.$$
Writing elements of the root string in terms of their coordinates on the root lattice $Q_0$, it is clear that they are all elements of $\Z\Pi(B)$. For example
$-p\a+\beta=\sum_{i\in K}(-pa_i+b_i)\alpha_i.$
\end{proof}


\subsection{Classification of $\Phi$-subsystems in rank 2}
Our next step is to classify the $\Phi$-subsystems in any infinite rank 2 root system using explicit formulas for the Weyl group reflections. Let $\Delta$ be a root system of type $H(a,b)$ for $a\ge b$ and $ab\ge4$.
Let $\Gamma\subseteq \Delta^\text{re}$ be nonempty.

First we note that $\Phi(\Gamma)$ is closed under negation, since $w_\a\a=-\a$. So, using the formulas of Lemma~\ref{L-neg},
$$\Phi(\Gamma)=\{\alpha^{LL}_j,\alpha^{LU}_{-j-1},\alpha^{SU}_k,\alpha^{SL}_{-k-1} \mid j\in I^L,\, k\in I^S\},$$
for some index sets $I^L,I^S\subseteq \mathbb{Z}$.
Every  real root has the form $\a=w\a_i$ for $i=1,2$ and $w\in W$, so the reflection in $\alpha$ is $w_\a=ww_iw^{-1}$.
We obtain the following formulas for the reflections corresponding to each real root:
$$w^{LL}_j = w^{LU}_{-j-1}= (w_1w_2)^{2j}w_1, \qquad  w^{SU}_j=w^{SL}_{-j-1}=(w_2w_1)^{2j}w_2.$$
We can use this to easily prove formulas for the action of a reflection on a real root:
\begin{lemma}\label{reflact} For all $j,k\in\mathbb{Z}$,
\begin{align*}
w^{LL}_k\a^{LL}_j&=-\a^{LL}_{2k-j},& w^{SU}_k\a^{SU}_j&= -\a^{SU}_{2k-j},\\
w^{LL}_k\a^{SU}_j&=-\a^{SU}_{-2k-j-1},& w^{SU}_k\a^{LL}_j&= -\a^{LL}_{-2k-j-1}.
\end{align*}
\end{lemma}

\begin{lemma} \label{reflincl} Given integers $j$ and $k$:
\begin{enumerate} 
\item\label{reflincl-l} If $j,k\in I^L$, then $j+(k-i)\mathbb{Z}\subseteq I^L$.
\item\label{reflincl-s} If $j,k\in I^S$, then $j+(k-j)\mathbb{Z}\subseteq I^S$.
\item\label{reflincl-ls}If $j\in I^L$, $k\in I^S$, then $j+(2j+2k+1)\mathbb{Z}\subseteq I^L$ and $k+(2j+2k+1)\mathbb{Z}\subseteq I^S$.
\end{enumerate}
\end{lemma}
\begin{proof}
Suppose $I^S$ contains $\ell:=j+(n-1)(k-j)$ and $m:=j+n(k-j)$. Then
Lemma~\ref{reflact} shows that $j+(n+1)(k-j)=2\ell-m\in I^S$ and
$j+(n-2)(k-j)=2m-\ell \in I^S$. Part (i) now follows by bidirectional induction. Part (ii) is similar.

Let $d:=2j+2k+1$. Now suppose $j,j+nd\in I^L$ and $k, k+nd\in I^S$.
Then $j-(n+1)d= -2k-(j+nd)-1\in I^L$ and so $j+(n+1)d =2j-(j-(n+1)d)\in I^S$.
Similar arguments show that $j+(n-1)d\in I^S$ and $k+(n\pm1)d\in I^L$.
Part (iii) now follows by bidirectional induction. 
\end{proof}

We can now classify the $\Phi$-subsystems in terms of their index sets:
\begin{proposition} $\;$ \label{Iclass}
\begin{enumerate} 
\item If $I^S$ is empty, then $I^L=r+d\mathbb{Z}$ for some $r,d\in\mathbb{Z}$ with $d\ge0$ and $0\le r<d$.
\item If $I^L$ is empty, then $I^S=r+d\mathbb{Z}$ for some $r,d\in\mathbb{Z}$ with $d\ge0$ and $0\le r<d$.
\item Otherwise, $I^L=r+(2d+1)\mathbb{Z}$ and $I^S=d-r+(2d+1)\mathbb{Z}$ for some $d\ge0$ and $d\le r\le d$.
\end{enumerate}
\end{proposition}
\begin{proof}
(i) Suppose $I^S$ empty and let $J=\{j\in\mathbb{Z}\mid \a^{LL}_j\in\Gamma\text{ or }\a^{LU}_{-j-1}\in\Gamma\}$, so
$\Phi(\Gamma)=\Phi(\{\a^{LL}_j\mid j \in J\})$.
If $J$ contains a single element, then take $r$ to be that element and $d=0$.
Otherwise, let $d$ be the greatest common divisor of all the integers $j-k$ for $j,k\in J$ with $j\ne k$. 
Let $r$ be the remainder of $j\in J$ divided by $d$, which is the same for all $j\in J$.
Then standard properties of integer lattices together with Lemma~\ref{reflincl}\ref{reflincl-l} show that 
$$J \subseteq r+d\mathbb{Z}\subseteq I^L.$$
It now suffices to show that $\{\alpha^{LL}_j,\alpha^{LU}_{-j-1} \mid j\in r+d\mathbb{Z}\}$ is a $\Phi$-subsystem, 
but this follows immediately from Lemmas~\ref{L-neg} and~\ref{reflact}.
The proof of (ii) is similar to (i).

(iii) 
The orbits of $W$ on  $\Delta^\re$ are $W\a_1$ and $W\a_2$, so 
$\Phi(\Gamma)\cap W\a_1=\{\alpha^{LL}_j,\alpha^{LU}_{-j-1} \mid i\in I^L\}$ and
$\Phi(\Gamma)\cap W\a_2=\{\alpha^{SU}_j,\alpha^{SL}_{-j-1}\mid j\in I^S\}$ are both $\Phi$-subsystems in their own rights.
By (i) and (ii), $I^L=r_1+d_1\mathbb{Z}$ and $I^S=r_2+d_2\mathbb{Z}$ for some  $d_i\ge0$, $0\le r_i<d_i$, for $i=1,2$.
For every $m\in\mathbb{Z}$, we have $r_1\in I^L$ and $r_2+md_2\in I^S$, so
Lemma~\ref{reflincl}\ref{reflincl-ls} implies that 
$r_1+(2r_1+2r_2+2md_2+1)m\mathbb{Z}\subseteq r_1+d_1\mathbb{Z}$.
Hence
$$d_1\mid (2r_1+2r_2+1)+2md_2,\qquad\text{for all $m\in\mathbb{Z}$}.$$
So $d_1\mid 2r_1+2r_2+1$ and hence $d_1$ is odd, say $d_1=2d+1$.
Also $d_1\mid 2d_2$ and hence $d_1\mid d_2$.
Reversing the roles of $I^L$ and $I^S$ we also get $d_2\mid d_1$,
so $d_1=d_2=2d+1$. We can choose $r$ such that $r\equiv r_1\pmod{2d+1}$ and $-d\le r\le d$, so that $I^L=r+(2d+1)\mathbb{Z}$.
Finally $2r+2r_2+1\equiv0\pmod{2d+1}$,
so 
$$r_2\equiv r_2 +2dr+2dr_2+d \equiv r_2-r-r_2+d\equiv d-r\pmod{2d+1},$$
and hence $I^S=d-r+(2d+1)\mathbb{Z}$.
\end{proof}

\begin{thm}\label{subsys}
Let $\Delta$ be an infinite rank 2 root system of type $H(a,b)$ with $a\ge b$ and $ab\ge4$.
Every nonempty $\Phi$-subsystem of $\Delta$ has simple roots, Cartan matrix, and inner product matrix given by one of the rows in Table~\ref{T-Phisubs}
where $\delta_d:= \eta_{d}-\eta_{d-1}$ and  $\epsilon_d:= \gamma_{d+1}-\gamma_d$.
In particular  all $\Phi$-subsystems of $\Delta$ have rank at most $2$.
\end{thm}
\begin{table}
\begin{tabular}{l||l|l|l|l}
Type & Integer conditions & Simple roots & Cartan Matrix & Inner product matrix\\\hline
$\text{I}_L$&$r$ arbitrary & $\a^{LL}_r$ & $A_1$ & $\frac{a}{b}A_1$\\
$\text{I}_S$&$r$ arbitrary & $\a^{SU}_r$ & $A_1$ & $A_1$\\
$\text{II}_L$ & $d>0$, $0\le r<d$ & $\a^{LL}_r,\a^{LU}_{d-r-1}$& $H(\delta_d,\delta_d)$& $\frac{a}{b} H(\delta_d,\delta_d)$\\
$\text{II}_S$&$d>0$, $0\le r<d$ & $\a^{SU}_r,\a^{SL}_{d-r-1}$& $H(\delta_d,\delta_d)$& $H(\delta_d,\delta_d)$\\
$\text{II}_{LS}$&$d\ge 0$, $-d\le r\le d$ & $\a^{LL}_r,\a^{SU}_{d-r}$& $H(a\epsilon_d, b\epsilon_d)$ & $B(a\epsilon_d, b\epsilon_d)$
\end{tabular}
\vspace{2mm}
\caption{$\Phi$-subsystems of rank 2 root systems}\label{T-Phisubs}
\end{table}
\begin{proof}
Let $\Phi'$ be a $\Phi$-subsystem of $\Delta$.
First suppose that $\Phi'\subseteq W\a_1$. 
Then Proposition~\ref{Iclass}(i) implies that 
$\Phi'=\{\alpha^{LL}_j,\alpha^{LU}_{-j-1} \mid j\in r+d\mathbb{Z}\},$
for some $d\ge0$ and $0\le r<d$.
If $d=0$, this gives us type $\text{I}_L$.
Otherwise it is easily shown that every positive root in $\Phi'$ is a
positive linear combination of $\a^{LL}_r$ and $\a^{LU}_{d-r-1}$, so this forms a base. The Cartan matrix and inner product matrix can be computed directly from the base. 
For example, if the Cartan matrix is $(c_{ij})$ then
\begin{align*}
c_{12} &= \frac{2(\a^{LL}_r,\a^{LU}_{d-r-1})}{(\alpha_r^{LL},\alpha^{LL}_r)}
  = \frac{b}{a} (\a^{LL}_r,\a^{LU}_{d-r-1})
  =\frac{b}{a} \left((w_1w_2)^r\a_1,(w_1w_2)^r\a^{LU}_{d-1}\right)
  =\frac{b}{a} (\a_1,\a^{LU}_{d-1})\\
  &= \frac{b}{a}\left(\begin{matrix}1&0\end{matrix}\right) \left(\begin{matrix}2a/b&-a\\-a&2\end{matrix}\right) 
  \left(\begin{matrix}\eta_{d-1}\\a\gamma _{d}\end{matrix}\right)
   = \frac{b}{a}\left(2\frac{a}{b}\eta_{d-1}-a^2\gamma_d\right)
 = 2\eta_{d-1}-ab\gamma_d
  = \eta_{d-1}-\eta_d =-\delta_d,
\end{align*} 
where the second last equality follows from Lemma~\ref{acp}(ii).
This gives type $\text{II}_L$.

Similarly we get types $\text{I}_S$ and $\text{II}_S$ from  Proposition~\ref{Iclass}(ii), and type $\text{II}_{LS}$ from Proposition~\ref{Iclass}(iii).
\end{proof}

%
%

Values of $\delta_d$ and $\epsilon_d$ for small $d$ are given in Table~\ref{T-subparams}.
\begin{table}
{\tiny
$$\begin{array}{r|rr}
d & \delta_d=\eta_d-\eta_{d-1} & \epsilon_d=\gamma_{d+1}-\gamma_{d} \\\hline
0 & & 1\\
1 & ab - 2 & ab - 3 \\
2 & a^2b^2 - 4ab + 2 & a^2b^2 - 5ab + 5 \\
3 & a^3b^3 - 6a^2b^2 + 9ab - 2 & a^3b^3 - 7a^2b^2 + 14ab - 7 \\
4 & a^4b^4 - 8a^3b^3 + 20a^2b^2 - 16ab + 2 & a^4b^4 - 9a^3b^3 +
27a^2b^2 - 30ab + 9 \\
5 & a^5b^5 - 10a^4b^4 + 35a^3b^3 - 50a^2b^2 + 25ab - 2 & a^5b^5 -
11a^4b^4 + 44a^3b^3 - 77a^2b^2 + 55ab - 11 \\
6 & a^6b^6 - 12a^5b^5 + 54a^4b^4 - 112a^3b^3 + 105a^2b^2 - 36ab + 2 &
a^6b^6 - 13a^5b^5 + 65a^4b^4 - 156a^3b^3 + 182a^2b^2 - 91ab + 13
\end{array}$$}
\caption{Values of $\delta_d$ and $\epsilon_d$ for small $d$}\label{T-subparams}
\end{table}

We can also classify of $\Phi$-subsystems as finite, affine or hyperbolic
systems:
\begin{thm} 
Let $\Delta$ be a rank 2 root system and let $\Gamma$ be a nonempty set of real roots in $\Delta$.
\begin{enumerate}
\item If $\Delta$ is finite, then $\Phi(\Gamma)$ is finite.
\item If $\Delta$ is affine of type $\widetilde{A}_1$, then $\Phi(\Gamma)$ has finite type $A_1$ or affine type $\widetilde{A}_1$.
\item If $\Delta$ is affine of type $\widetilde{A}_2^{(2)}$, then $\Phi(\Gamma)$ has finite type $A_1$, or affine type $\widetilde{A}_1$ or $\widetilde{A}_2^{(2)}$.
\item If $\Delta$ is hyperbolic, then $\Phi(\Gamma)$ has finite type $A_1$ or hyperbolic type.
\end{enumerate}
\end{thm}
\begin{proof}
Part (i) is clear. The finite type $A_1$ occurs exactly when $\Gamma\subseteq\{\pm\alpha\}$, so we will assume from now on that this is not the case.

If $\Delta$ is affine, then $ab=4$ and it is easy to show from the recursion formulas in Lemma~\ref{acp} that
$\delta_d=\eta_d-\eta_{d-1}=2$ and $\epsilon_d=\gamma_{d+1}-\gamma_d=1$.
Parts (ii) and (iii) now follow.

If $\Delta$ hyperbolic, then $ab>4$, and so for $d>1$
$$\delta_d=\eta_d-\eta_{d-1} = (ab-2)\eta_{d-1} -\eta_{d-2}-\eta_{d-1} >2\eta_{d-1} -\eta_{d-2}-\eta_{d-1} =\eta_{d-1}-\eta_{d-2}=\delta_{d-1}.$$
By induction we get $\delta_d \ge \delta_1=(ab-1)-1=ab-2>2$ for all $d>0$.
It now follows that $H(\delta_d,,\delta_d)$ is hyperbolic since ${\delta_d}^2>4$.

A similar argument shows that
$\epsilon_d>\epsilon_{d-1}$ for $d>0$,
and so $\epsilon_d\ge\epsilon_0=1$ for $d\ge0$.
and so $H(a\epsilon_d,b\epsilon_d)$ is hyperbolic.
\end{proof}
As part of  the last proof we showed that the sequences $\delta_d$ and $\epsilon_d$ are strictly increasing when $\Delta$ is hyperbolic, so Theorem~1.2 is now proved for $\Phi$-subsystems.

%
%

We now consider the classification of $\Delta$-subsystems of $\Delta$.
Let $\Gamma\subseteq \Delta^\re$ nonempty and recall that $\Delta(\Gamma)=\mathbb{Z}\Gamma\cap\Delta$,
 $\Delta^\re(\Gamma)=\mathbb{Z}\Gamma\cap\Delta^\re$.
Since the imaginary roots of an affine or hyperbolic root system are just the linear combinations of real roots
with nonpositive norm, it will suffice to describe $\Delta^\re(\Gamma)$.
From the definition of a reflection, we can see that $w_\a\realroots(\Gamma)\subseteq \realroots(\Gamma)$ for all $\a\in\Gamma$, and so
$$\Phi(\Gamma)\subseteq\realroots(\Gamma).$$
We also have $\Phi(\realroots(\Gamma))=\realroots(\Gamma)$, so the real roots of a $\Delta$-subsystem always form a $\Phi$-subsystem, but possibly for a dif and only iferent set of generators.
The classification of $\Delta$ subsystems reduces to divisibility properties for the sequences $\eta_j$ and $\gamma_j$.
\begin{lemma}\label{divrec} Let $a\ge b\ge1$ with $ab\ge4$, and let $d\ge0$, $i\in\mathbb{Z}$. Then
\begin{align}
\label{divrec-gg} \gamma_d\delta_{j-d}&=\gamma_j-\gamma_{j-2d},\\
\label{divrec-eg} \eta_d\epsilon_{j-d-1}&=\gamma_j-\gamma_{j-2d-1},\\
\label{divrec-ee} \eta_d\delta_{j-d}&=\eta_j-\eta_{j-2d-1},\\
\label{divrec-ge} ab\gamma_d\epsilon_{j-d}&=\eta_j-\eta_{j-2d}.
\end{align}
\end{lemma}
\begin{proof}
The equations are easy to prove for $d=0$.
Note that 
\begin{align*}
\delta_{j-1}&=\eta_{j-1}-\eta_{j-2}=(ab\gamma_{j}-\eta_j)-(ab\gamma_{j-1}-\eta_{j-1})=ab\epsilon_{j-1}-\delta_j,\quad\text{and}\\
\epsilon_{j-1}&=\gamma_{j}-\gamma_{j-1}=(\eta_{j}-\gamma_{j+1})-(\eta_{j-1}-\gamma_j)=\delta_{j}-\epsilon_j.
\end{align*}
Assume all of the equations hold for $d\le e$. First we prove \eqref{divrec-gg} and \eqref{divrec-ge} for $d=e+1$:
\begin{align*}
\gamma_{e+1}\delta_{j-e-1}
&=(\eta_e-\gamma_e)(ab\epsilon_{j-e-1}-\delta_{j-e})\\
&=ab\eta_e\epsilon_{j-e-1} -ab\gamma_e\epsilon_{j-e-1}-\eta_e\delta_{j-e}+\gamma_e\delta_{j-e}\\
&=ab(\gamma_j- \gamma_{j-2e-1})-(\eta_{j-1}-\eta_{j-2e})-(\eta_j-\eta_{j-2e-1})+(\gamma_j-\gamma_{j-2e})\\\
&=\gamma_j + (ab\gamma_j-\eta_{j-1}-\eta_j) +(\eta_{j-2e-1}-ab\gamma_{j-2e-1})+(\eta_{j-2e}-\gamma_{j-2e})\\
&=\gamma_j+0-\eta_{j-2e-2}-\gamma_{j-2e-1}=\gamma_j-\gamma_{j-2e-2},\\
ab\gamma_{e+1}\epsilon_{j-e-1}&=ab(\eta_e-\gamma_e)(\delta_{j-e}-\epsilon_{j-e})\\
&=ab\eta_e\delta_{j-e} -ab\gamma_e\delta_{j-e}-ab\eta_e\epsilon_{j-e}+ab\gamma_e\epsilon_{j-e}\\
&=ab(\eta_j- \eta_{j-2e-1})-ab(\gamma_{j}-\gamma_{j-2e})-ab(\gamma_{j+1}-\gamma_{j-2e})+(\eta_j-\eta_{j-2e})\\\
&=\eta_j +ab(\eta_j-\gamma_{j}-\gamma_{j+1}) -ab(\eta_{j-2e-1}-\gamma_{j-2e})+(ab\gamma_{j-2e}-\eta_{j-2e})\\
&=\eta_j+0- ab\gamma_{j-2e-1}+\eta_{j-2e-1}=\eta_j-\eta_{j-2e-2}.
\end{align*}
The proof of  \eqref{divrec-eg} and \eqref{divrec-ee} for $d=e+1$ is similar. Equations (1)-(4) follow by induction for $d\ge0$.
\end{proof}

\begin{lemma} \label{div} Let $a\ge b\ge1$ with $ab\ge4$, and let $d\ge0$, $j\in\mathbb{Z}$.
\begin{enumerate}
\item\label{div-ab} $\gcd(a,\eta_j)=\gcd(b,\eta_j)=1$. 
\item\label{div-gg} $\gamma_d \mid \gamma_j$ if and only if $j\in d\mathbb{Z}$.
\item\label{div-eg} $\eta_d\mid \gamma_j$ if and only if $j\in(2d+1)\mathbb{Z}$. 
\item\label{div-ee} $\eta_d\mid\eta_j$ if and only if $j\in d+(2d+1)\mathbb{Z}$.
\item\label{div-ge} $\gamma_d\mid\eta_j$ if and only if $d=1$, when $ab>4$.
\item\label{div-ge41} $\gamma_d\mid\eta_j$ if and only if $d=2e+1$ is odd and $j\in e+(2e+1)\mathbb{Z}$, when $ab=4$.
\end{enumerate}
\end{lemma}

\begin{proof}
\ref{div-ab} This follows from the fact that $\eta_j\equiv(-1)^j\pmod{ab}$, which is easily proved by induction.

Cases (ii)-(v) proceed by repeated application of Lemma~\ref{divrec}. Part (v) also uses Lemma~\ref{acp}. Part (vi) is immediate.
\end{proof}


\newpage
\begin{thm}
\label{subsysbeq1}
Let $\Delta$ be a rank 2 root system of type $H(a,b)$ with $a\ge b$ and $ab\ge4$.
Let $\Gamma\subseteq\Delta^\re$ be nonempty. 

\begin{enumerate}
\item If $a>4$, $b=1$ and $\Phi(\Gamma)$ is the subsystem consisting of all short roots in $\Delta^\re$, then $\Delta^\re(\Gamma)=\Delta^\re\ne\Phi(\Gamma)$.
 

\item If $a=4$, $b=1$ and $\Phi(\Gamma)$ is a subsystem of type $\text{II}_S$ with base $\a^{SU}_r,\a^{SL}_{d-r-1}$ for some odd $d=2e+1$ and $0\le r<d$,
then $\Delta^\re(\Gamma)\ne\Phi(\Gamma)$ is a subsystem of type $\text{II}_{LS}$ with base $\a^{LL}_{s},\a^{SU}_{e-s}$ where $s\equiv e-r\pmod{d}$ and $-e\le s\le e$. 

\item In all other cases,  $\Delta^\re(\Gamma)=\Phi(\Gamma)$.
\end{enumerate}
\end{thm}

\medskip
\begin{proof} 
If $\Phi(\Gamma)$ has type $\text{I}_L$ or $\text{I}_S$, then it is clear that  $\Delta^\re(\Gamma)=\Phi(\Gamma)$.\\

Suppose $\Phi(\Gamma)$ has type $\text{II}_L$. Since $\Phi(\Gamma)=(w_1w_2)^r\Phi(\{\a^{LL}_0,\a^{LU}_{d-1}\})$, it suffices to consider $r=0$.\\

Now $\a^{LL}_0=\a_1$ and $\a^{LU}_{d-1}=\eta_{d-1}\a_1+a\gamma_d\a_2$, so
$$\Delta^\re(\Gamma) =  \mathbb{Z}\{\alpha^{LL}_0,\alpha^{LU}_{d-1}\} \cap \Delta^\re = \mathbb{Z}\{\alpha_1,a\gamma_d\alpha_2\}\cap \Delta^\re.$$

The root $\a^{LL}_j=\eta_{j}\alpha_1+a\gamma_{j}\alpha_2$ is in $\Delta^\re(\Gamma)$ if and only if $\gamma_d\mid\gamma_j$ if and only if $j\in d\mathbb{Z}$ by Lemma~\ref{div}\ref{div-gg}.\\

Also $\a^{SU}_j=b\gamma_j\a_1+\eta_i\a_2$ is in $\Delta^\re(\Gamma)$ if and only if $a\gamma_d\mid \eta_j$ which is not possible by Lemma~\ref{div}\ref{div-ab} since $a>1$.
Hence $\Delta^\re(\Gamma)=\Phi(\Gamma)$.\\

Suppose $\Phi(\Gamma)$ has type $\text{II}_{LS}$. Since $\Phi(\Gamma)=(w_2w_1)^{d-r}\Phi(\{\a^{LL}_d,\a^{SU}_{0}\})$, it suffices to consider $r=d$.

We have $\a^{LL}_{d}=\eta_d\a_1+a\eta_d\a_2$ and $\a^{SU}_0=\a_2$, so
$$\Delta^\re(\Gamma) =  \mathbb{Z}\{\alpha^{LL}_d,\alpha^{SU}_{0}\} \cap \Delta^\re = \mathbb{Z}\{\eta_d\alpha_1,\alpha_2\}\cap \Delta^\re.$$

Now $\a^{LL}_j=\eta_{j}\alpha_1+a\gamma_{j}\alpha_2$ is in $\Delta^\re(\Gamma)$ if and only if $\eta_d\mid\eta_j$  if and only if $j\in d+(2d+1)\mathbb{Z}$ by Lemma~\ref{div}\ref{div-ee}.
And $\a^{SU}_j=b\gamma_j\a_1+\eta_j\a_2$ is in $\Delta^\re(\Gamma)$ if and only if $\eta_d\mid b\gamma_i$ if and only if $j\in (2d+1)\mathbb{Z}$ by Lemma~\ref{div}\ref{div-eg}.
Hence $\Delta^\re(\Gamma)=\Phi(\Gamma)$.\\

Finally suppose $\Phi(A)$ has type $\text{II}_S$. Since $\Phi(\Gamma)=(w_2w_1)^r\Phi(\{\a^{SU}_0,\a^{SL}_{d-1}\})$, it suffices to consider $r=0$.\\

Now $\a^{SU}_0=\a_2$ and $\a^{SL}_{d-1}=b\gamma_d\a_1+\eta_{d-1}\a_2$, so
$$\Delta^\re(\Gamma) =  \mathbb{Z}\{\alpha^{LL}_0,\alpha^{LU}_{d-1}\} \cap \Delta^\re = \mathbb{Z}\{b\gamma_d\alpha_1,\alpha_2\}\cap \Delta^\re.$$

We have $\a^{SU}_j=b\gamma_i\a_1+\eta_j\a_2\in \Delta^\re(\Gamma)$ if and only if $\gamma_d\mid \gamma_j$ if and only if $j\in d\mathbb{Z}$.
And $\a^{LL}_j=\eta_{j}\alpha_1+a\gamma_{j}\alpha_2$ is in $\Delta^\re(\Gamma)$ if and only if $b\gamma_d\mid\eta_j$.
By Lemma~\ref{div}\ref{div-ab}, \ref{div-ge}, and \ref{div-ge41},  this can only happen if $a>4$, $b=1$ and $d=1$; or $a=4$, $b=1$ and $d$ odd.\\

If $a>4$, $b=1$, and $d=1$, then $\Phi(\Gamma)$ is the set of all short real roots, and $b\gamma_d\mid\eta_j$ for all $j$ so $\Delta^\re(\Gamma)=\Delta^\re$.\\

If $a=4$, $b=1$, and $d=2e+1$, then $b\gamma_d\mid\eta_i$ if and only if $j\in e+(2e+1)\mathbb{Z}$, so $I^S=(2e+1)\mathbb{Z}$, $I^L=e+(2e+1)\mathbb{Z}$, and hence $\Delta^\re(\Gamma)$ has type $\text{II}_{LS}$ with the given basis. In all other cases  $\Delta^\re(\Gamma)=\Phi(\Gamma)$.
\end{proof}

Theorem~1.1 and Theorem~1.2 for $\Delta$-subsystems now follow.

\end{document}